\documentclass[10pt,final,journal,twocolumn]{IEEEtran}
\usepackage{amsmath,amsfonts,amsthm,amssymb}
\usepackage{setspace}
\usepackage{fancyhdr}
\usepackage{extramarks}
\usepackage{chngpage}
\usepackage[usenames,dvipsnames]{color}
\usepackage{ifthen}
\usepackage{courier}
\usepackage{graphicx}
\usepackage{textcomp}
\usepackage{amssymb}
\usepackage{cancel}
\usepackage{mathrsfs}
\usepackage{hyperref}
\usepackage{array}
\usepackage{relsize}
\usepackage{multirow,bigdelim}
\usepackage{blkarray}
\usepackage[nospace]{cite}
\usepackage{algorithm}
\usepackage[noend]{algpseudocode}\usepackage{subfigure}

\newcommand\numberthis{\addtocounter{equation}{1}\tag{\theequation}}

\newtheorem{theo}{\textbf{Theorem}}
\newtheorem{lem}{\textbf{Lemma}}

\theoremstyle{definition}
\newtheorem{defi}{\textbf{Definition}}

\newtheorem*{remk}{\textit{Remark}}

\allowdisplaybreaks

\makeatletter
\def\BState{\State\hskip-\ALG@thistlm}
\makeatother

\begin{document}

\title{{ITLinQ: A New Approach for Spectrum Sharing in Device-to-Device Communication Systems}}

\author{Navid~Naderializadeh and~A.~Salman~Avestimehr%
\thanks{Manuscript received December 1, 2013;  revised April 4, 2014 and May 24, 2014. The research of A. S. Avestimehr and N. Naderializadeh is supported by NSF Grants CAREER 1408639, CCF-1408755, NETS-1419632, EARS-1411244, ONR award N000141310094, AFOSR YIP award FA9550-11-1-0064, research grants from Intel and Verizon via the 5G project, and a gift from Qualcomm.

This work has been presented in part at the IEEE International Symposium on Dynamic Spectrum Access Networks (DySPAN), McLean, VA, April 2014 \cite{dyspan} and will also be presented in part at the IEEE International Symposium on Information Theory (ISIT), Honolulu, HI, July 2014 \cite{isit}.

The authors are with the Department of Electrical Engineering, University of Southern California, Los Angeles, CA 90089, USA (emails: naderial@usc.edu, avestimehr@ee.usc.edu).}}

\maketitle
\begin{abstract}
We consider the problem of spectrum sharing in device-to-device communication systems. Inspired by the recent optimality condition for treating interference as noise, we define a new concept of \emph{information-theoretic independent sets} (ITIS), which indicates the sets of links for which simultaneous communication and treating the interference from each other as noise is information-theoretically optimal (to within a constant gap). Based on this concept, we develop a new spectrum sharing mechanism, called \emph{information-theoretic link scheduling} (ITLinQ), which at each time schedules those links that form an ITIS. We first provide a performance guarantee for ITLinQ by characterizing the fraction of the capacity region that it can achieve in a network with sources and destinations located randomly within a fixed area. Furthermore, we demonstrate how ITLinQ can be implemented in a distributed manner, using an initial 2-phase signaling mechanism which provides the required channel state information at all the links. Through numerical analysis, we show that distributed ITLinQ can outperform similar state-of-the-art spectrum sharing mechanisms, such as FlashLinQ, by more than a 100\% of sum-rate gain, while keeping the complexity at the same level. Finally, we discuss a variation of the distributed ITLinQ scheme which can also guarantee fairness among the links in the network and numerically evaluate its performance.
\end{abstract}


\section{Introduction}\label{intro}

Device-to-device (D2D) communication among mobile users is drawing considerable attention for the development of next-generation wireless communication systems. The D2D communication functionality can enable various applications and services (see, e.g. \cite{ericsn_5g,d2djef}), such as proximity-based applications involving discovering and communicating with nearby devices (e.g., Internet of Things). Such functionality can also enable higher data rates and system capacity by leveraging the underlying peer-to-peer wireless network that can be created via local communication among the users (see, e.g. \cite{d2dlte1,d2ddopp,fodor,d2dinterf,d2dicufn}). 
Moreover, incorporating caching capability into D2D communication networks have been shown to also significantly enhance the system throughput for applications that follow a popularity pattern, such as the on-demand video traffic for which a few dominant videos account for a large part of the traffic~\cite{caching}.

However, considering the increasing density of mobile users in wireless networks, the problem of spectrum sharing and interference management inside D2D communication networks becomes of vital importance for the aforementioned applications and improvements. The main challenge for interference management in such networks is that neither fully coordinated synchronous cellular-type approaches that rely on advanced physical layer designs, nor fully distributed and asynchronous WiFi-type mechanisms (such as CSMA/CA) are adequate. The downside of the first type of interference management mechanisms is that they need levels of centralization, coordination, and information at the mobile nodes that are difficult to accomplish in practice. On the other hand, the problem with the second type of approaches is that their performance degrades significantly as the number of links grows.

These issues have motivated a more recent approach that is based on a minimal level of coordination among the links which also maintains its promising performance for large numbers of links. This scheme, called FlashLinQ \cite{FLQ}, is a distributed scheduling scheme which demonstrates considerable improvement over pure CSMA/CA. In a system of multiple source-destination pairs (links), this scheduling algorithm first orders the links according to a randomly selected priority list. Then, starting from the higher-order links, each link is scheduled if it does not cause and does not receive ``much'' interference from the already scheduled links. The level of acceptable interference is determined based on the observed signal-to-interference ratio (SIR) at all the previously scheduled links \emph{and} also the current link. FlashLinQ has also been implemented and experimented in practice and shown to demonstrate promising performance compared to previous scheduling schemes.

FlashLinQ scheduling can also be viewed as a refinement of the conventional independent set scheduling which is based on using a conflict graph to model the interference among the links (see, e.g. \cite{is1,is2,is3,is4,is5} and the protocol model in \cite{kumar}). In the independent set scheduling approach, two links (source-destination pairs) are considered to be mutually non-interfering, hence able to transmit data at the same time, if the interference that they cause on each others' destinations is below a certain threshold. The drawback of this scheme is that this threshold is set at a fixed value (often at noise level) which does not capture the effect of the number of links, their density inside the cell area, etc. More importantly, the scheme does not consider the signal-to-noise ratio (SNR) level that each link itself can achieve and only takes the interference levels into account. FlashLinQ, however, overcomes this problem by comparing the direct signal power level that each link gets with its incoming interference power level. Also, in the FlashLinQ scheduling algorithm, if a link does not cause/receive much interference to/from higher-priority links, but does not get a high direct signal power itself, it gets silent and ``yields'' such that lower-priority links have the opportunity to contribute more to the overall sum-throughput of the network.



Hence, both FlashLinQ and independent set scheduling approaches aim at finding subsets of links in which the interference among them is at a ``sufficiently'' low level, so that their simultaneous transmissions are not detrimental to each other. This gives rise to a natural question: what would be a theoretically-justified way of creating such subsets, and determining whether the interference among them is at a ``sufficiently'' low level? 

In this paper, we propose an answer to this question. We define a new concept of \emph{information-theoretic independent sets} (ITIS), which indicates the sets of links for which simultaneous communication and treating the interference from each other as noise is information-theoretically optimal (to within a constant gap). In other words, a subset of links forms an ITIS if by a simple scheme of using point-to-point Gaussian codebooks with appropriate power levels at each transmitter and treating interference as noise at every receiver we can achieve the entire information-theoretic capacity region of that subset of links (to within a constant gap). We use the recent optimality condition for treating interference as noise in \cite{tin} to provide a description of ITIS based on the channel gains among the links in the network. In fact, as we will see later, a subset of links is defined to create an ITIS if for any link in the subset, the SNR level is no less than the sum of its strongest incoming interference-to-noise ratio (INR) and its strongest outgoing INR (all measured in dB scale). It is important to note that this condition is quite different from that of FlashLinQ and independent set scheduling which only rely on thresholds on SIR and INR values to identify the subsets of links with ``sufficiently'' low levels of interference.

Furthermore, we propose our new spectrum sharing mechanism, named  \emph{information-theoretic link scheduling} (in short, ITLinQ), which schedules the links in an information-theoretic independent set to transmit data at the same time. We characterize the guaranteed fraction of the capacity region that ITLinQ is able to achieve in a specific network setting. In particular, we consider a set of $n$ source-destination pairs, where the source nodes are spread randomly and uniformly over a circular cell of fixed radius and each destination node is located within a distance $r_n\propto n^{-\beta}$ of its corresponding source node. For the channel gains, we only consider the path-loss model. In such a setting, we show that the criteria for defining information-theoretic independent sets transforms the network into a random geometric graph which enables us to characterize the fraction of capacity region that can be achieved by the ITLinQ scheme. In fact, depending on the value of $\beta$, we identify three regimes in each of which ITLinQ can achieve a fraction $\lambda$ of the capacity region within a gap of $k$ almost-surely:

\begin{itemize}
\item For $0<\beta<1$, $\lambda= \Theta\left( n^{\beta-1}\right)$ and $k=O\left(\frac{\log 3n}{n^{1-\beta}}\right)$.\footnote{For two functions $f(n)$ and $g(n)$ defined on the set of positive integers $\mathbb{Z}_+$, $f(n)=O(g(n))$ if and only if there exists a positive real number $a$ and a positive integer $n_0$ such that for all $n>n_0$, $|f(n)|\leq a |g(n)|$. Also, $f(n)=\Theta(g(n))$ if and only if there exist two positive real numbers $a_1$ and $a_2$ and a positive integer $n_0$ such that for all $n>n_0$, $a_1 |g(n)|\leq |f(n)|\leq a_2 |g(n)|$.}

\item For $\beta=1$, $\lambda= \frac{\ln \left({\ln n}\right)}{\ln n}$  and
\begin{align*}
k= O\left(\log (\ln n)+\frac{ \ln \left({\ln n}\right)}{\ln n}\right).
\end{align*}

\item For $\beta>1$, $\lambda= \Theta(1)$ and $k=O({\log 3n})$.
\end{itemize}

This shows a considerable improvement over the fraction of the capacity region that the conventional independent set scheduling can achieve, which is $\frac{1}{n}$ (derived via numerical analysis). We will also show that the network model in which each destination gets associated with the closest source to itself is a subclass of the above model for any $\beta<\frac{1}{2}$, and therefore we can asymptotically achieve a fraction $\Theta(\sqrt{n})$ of the capacity region in this case almost-surely.

Afterwards, we will focus on the challenge of distributed implementation of the ITLinQ scheme. We will develop a distributed spectrum sharing scheme based on ITLinQ, whose complexity is comparable to the FlashLinQ algorithm. The conditions that need to be satisfied for the sources and destinations in this scheme are based on the sufficient conditions for the optimality of treating interference as noise (as mentioned in \cite{tin}), hence providing a strong theoretical backbone for the algorithm. We will numerically evaluate the performance of our distributed scheme and compare it with FlashLinQ in an outdoor setting with 8 - 4096 links of random lengths spread uniformly at random in a square cell. We observe that the sum-rate achieved by the distributed ITLinQ scheme improves over that of FlashLinQ by more than 100\%, while keeping the complexity basically at the same level. Finally, we introduce a slight variation to the distributed ITLinQ scheme so as to consider fairness among the links in the network and show that our fair ITLinQ scheme achieves almost the same tail distribution for the link rates as in FlashLinQ, while demonstrating over 50\% sum-rate gain.

The rest of the paper is organized as follows. In Section \ref{desc_itq}, we formally describe the notion of information-theoretic independent sets and the ITLinQ scheme and present a capacity analysis of this scheme. In Section \ref{decent}, we will propose a distributed version of the ITLinQ scheme. In Section \ref{sim}, we numerically evaulate the performance of distributed ITLinQ and fair ITLinQ and compare them with FlashLinQ. We will conclude the paper in Section \ref{conc}.

\section{Description and Analysis of the Information-Theoretic Link Scheduling Scheme}\label{desc_itq}

In this section, we introduce our scheduling scheme, which we call ``information-theoretic link scheduling'' (in short, ``ITLinQ''). We start by defining the notion of ``information-theoretic independent set'' (in short, ``ITIS'') and then move forward to describe the ITLinQ scheme. Afterwards, we will consider a specific network setting and in that setting, we will characterize the fraction of capacity region that ITLinQ is able to achieve to within a gap.

\subsection{Description of ITIS and ITLinQ}\label{defs}

We consider a wireless network composed of $n$ sources $\{\text{S}_i\}_{i=1}^n$ and $n$ destinations $\{\text{D}_i\}_{i=1}^n$ in which each source aims to communicate a message to its corresponding destination. All the links (i.e., source-destination pairs) are considered to share the same spectrum, which gives rise to interference among all the transmissions. We assume that all the nodes (i.e., all the sources and the destinations) know how many links exist in the network and they also agree on a specific ordering of the links, where by ordering we mean a labeling of the links from 1 to $n$. Furthermore, we assume that the nodes are synchronous; i.e., there exists a common clock among them.

The physical-layer model of the network is considered to be the AWGN model in which each source $\text{S}_i$ intends to send a message $W_i$ to its corresponding destination $\text{D}_i$, and does so by encoding its message to a codeword $X_i^l$ of length $l$ and transmitting it within $l$ time slots. There is a power constraint of $\mathbb{E}\left[\frac{1}{l}\| X_i^l \|^2\right]\leq P$ on the transmit vectors. The received signal vector of destination $j$ over the $l$ time slots will be equal to
\begin{align*}
Y_j^l=\sum_{i=1}^n h_{ji} X_i^l + Z_j^l,
\end{align*}
where $h_{ji}$ denotes the channel gain between source $i$ and destination $j$, and $Z_j^l$ denotes the additive white Gaussian noise vector at destination $j$ with distribution $\mathcal{CN}(0,N\mathbf{I}_l)$, $\mathbf{I}_l$ being the $l \times l$ identity matrix. An example of such a network configuration is illustrated in Figure \ref{fig_network}.

\begin{figure}[h]
\centering
\includegraphics[trim = 0.4in .3in 0.2in .8in, clip,width=0.35\textwidth]{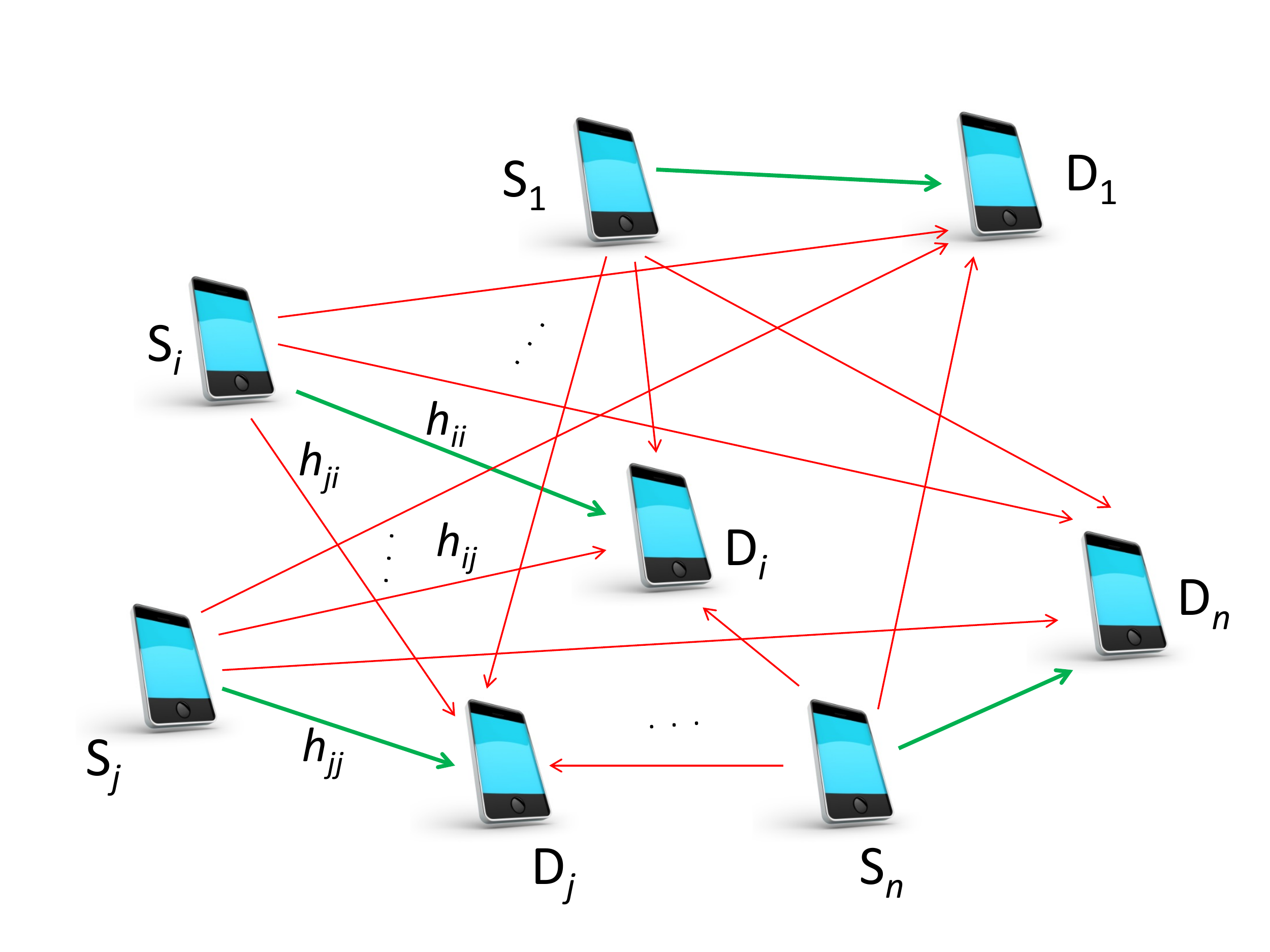}
\caption{A wireless network composed of $n$ source-destination pairs, where the green and red lines represent the direct and cross channel gains, respectively.}
\label{fig_network}
\end{figure}

We assume that at each destination, all the incoming interference is treated as noise. Therefore, each source-destination pair $\text{S}_i-\text{D}_i$ can achieve the rate of $R_i=\log(1+\text{SINR}_i)$, where $\text{SINR}_i\triangleq\frac{P|h_{ii}|^2}{\sum_{j\neq i} P |h_{ij}|^2+N}$ denotes the signal-to-interference-plus-noise ratio at destination $i$.

In general, treating interference as noise is known to be suboptimal for the general interference channel and numerous more sophisticated physical-layer coding schemes (such as message splitting and successive interference cancellation \cite{hk,etw}, interference alignment \cite{cadambe,maddahali}, and structured coding \cite{bpt,oen,jv}) have been proposed in order to improve it. However, the recent result in \cite{tin} proves that under a general condition in a network consisting of multiple source-destination pairs, treating interference as noise is information-theoretically optimal (to within a constant gap). The result is reflected in Theorem \ref{tin_thm}.

\begin{theo}[\hspace{-.001in}\cite{tin}]\label{tin_thm}
In a wireless network of $n$ source-destination pairs, if the following condition is satisfied, then treating interference as noise (in short, TIN) can achieve the whole capacity region to within a constant gap of $\log 3n$:
\begin{align}\label{tin_opt}
\emph{SNR}_i\geq \max_{j\neq i} \emph{INR}_{ij}  \max_{k\neq i} \emph{INR}_{ki},~~\forall i=1,...,n,
\end{align}
where $\emph{SNR}_i\triangleq\frac{P|h_{ii}|^2}{N}$ and $\emph{INR}_{ij}\triangleq\frac{P|h_{ij}|^2}{N}$ denote the signal-to-noise ratio of link $i$ and the interference-to-noise ratio of source $j$ at destination $i$, respectively.
\end{theo}

\begin{remk}
This result can intuitively be explained under the deterministic channel model of \cite{adt} as follows. Consider the deterministic model for link $i$ as shown in Figure \ref{tin_pic}.
%
\begin{figure}[h]
\centering
\includegraphics[trim = 1.4in 1.85in 1.6in 1.5in, clip,width=0.45\textwidth]{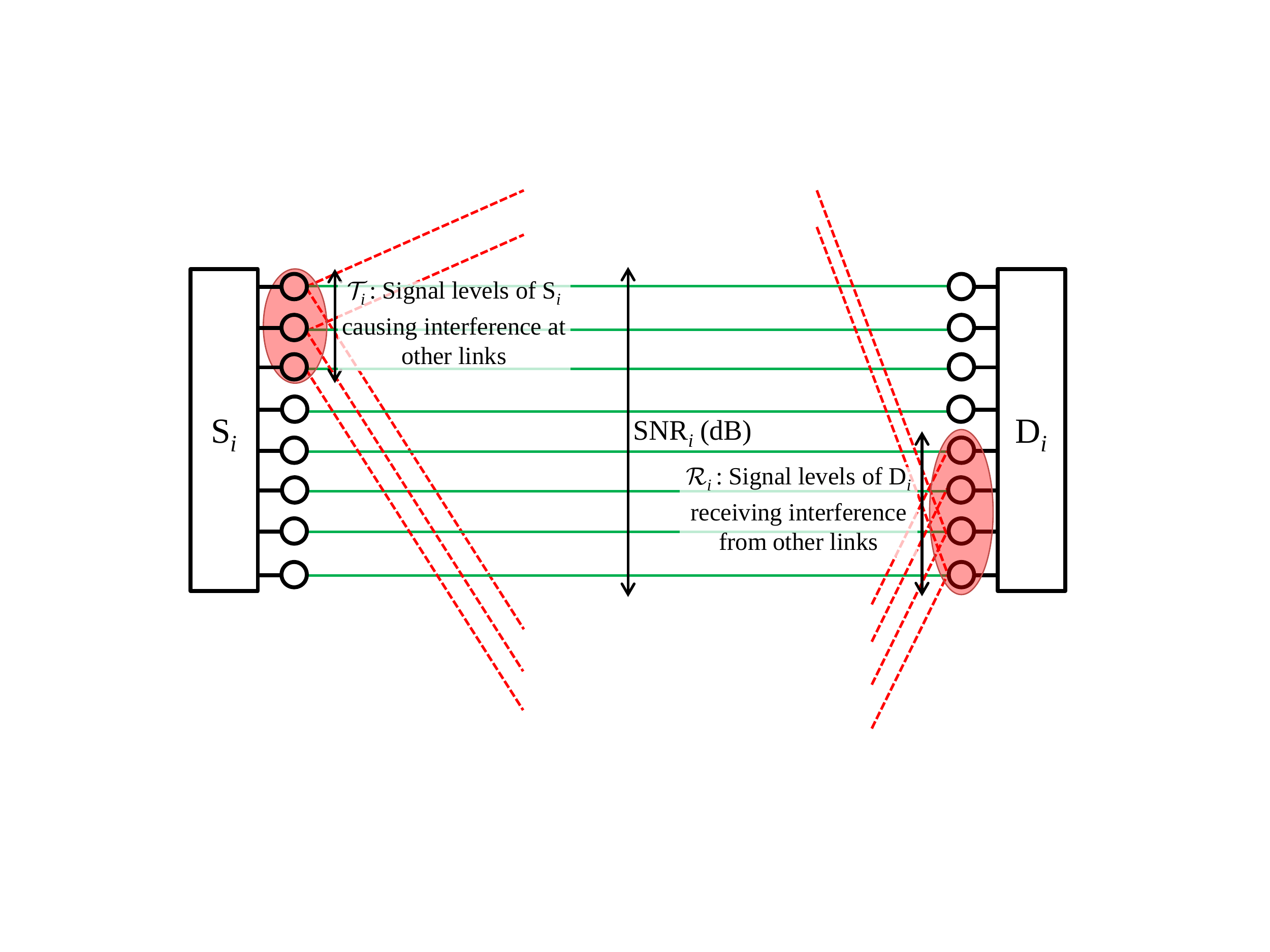}
\caption{A deterministic view of the optimality condition for treating interference as noise.}
\label{tin_pic}
\end{figure}
In this figure, each little circle represents a signal level. The transmit and received signal levels are sorted from MSB to LSB from top to bottom at the source and the destination respectively. The channel gain between source $i$ and destination $j$ in the deterministic model, denoted by $n_{ji}$ indicates how many of the first MSB transmitted signal levels of source $i$ are received at destination node $j$. Now, let us consider all transmit signal levels of source $i$ that are interfering to destinations $j \neq i$ in the network. The total number of them is $\max_{j \neq i} n_{ji}$, and they are depicted by the the set $\mathcal{T}_i$ in Figure \ref{tin_pic}. Similarly, we can consider all received signal levels of destination $i$ that are receiving interference from sources $j \neq i$ in the network. The total number of them is $\max_{j \neq i} n_{ij}$, and they are depicted by the the set $\mathcal{R}_i$ in Figure \ref{tin_pic}. Now, by taking logarithm of both sides of inequality (1) in Theorem 1, the condition in (1) can be described as $|\mathcal{T}_i|+|\mathcal{R}_i| \leq \log \text{SNR}_i=n_{ii}$, which means that there is no connection (or coupling) between the signal levels of source $i$ that are causing interference and the received signal levels of destination $i$ that are getting interference from all other nodes in the network.
\end{remk}

\begin{remk}
The gap of $\log 3n$ mentioned in Theorem \ref{tin_thm} is in fact the worst-case gap between the achievable rate region by TIN and the outer bound. We have numerically evaluated the actual gap for the case of randomly-generated networks with 8 links which satisfy condition (\ref{tin_opt}) and the result is illustrated in Figure \ref{actual_gap} in the form of the cumulative distribution function (CDF).
\begin{figure}[h]
\centering
\includegraphics[trim = 0.4in 2.7in .5in 2.8in, clip,width=0.45\textwidth]{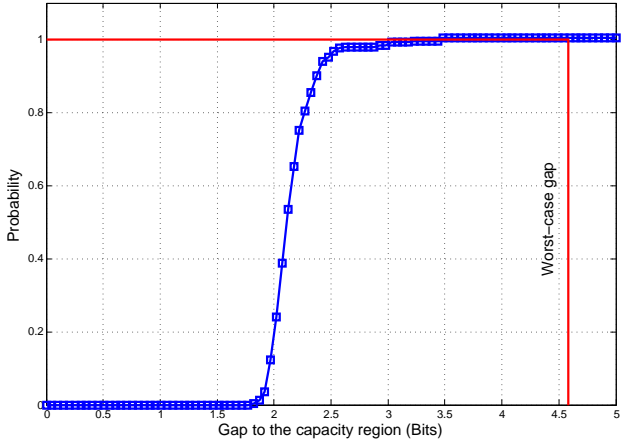}
\caption{Cumulative distribution function (CDF) of the actual gap of the achievable rate region by TIN with respect to the capacity region for networks with 8 links satisfying condition (\ref{tin_opt}) and its comparison with the worst-case gap of $\log 3n=\log 24\approx4.58$.}
\label{actual_gap}
\end{figure}
As the figure suggests, the actual gap to the capacity region achievable by TIN is much smaller than the worst-case gap of $\log 3n$ with a high probability.

\end{remk}

Therefore, if we consider any subset of the source-destination pairs in a wireless network and show that condition (\ref{tin_opt}) is satisfied in that subset, then we know that TIN is information-theoretically optimal in that subset of the links (to within a constant gap). This implies that the interference is at such a low level in this subnetwork that makes it suitable to call such a subset an ``information-theoretic independent subset''. More formally, we have the following definition.

\begin{defi}[ITIS]\label{itis}
In a wireless network of $n$ links, a subset of the links $\mathcal{S}\subseteq\{1,...,n\}$ is called an \emph{information-theoretic independent set} (in short, ITIS) if for any link $i\in\mathcal{S}$,
\begin{align}\label{itis_cond}
\text{SNR}_i\geq \max_{j\in\mathcal{S}\setminus \{i\}} \text{INR}_{ij}  \max_{k\in\mathcal{S}\setminus \{i\}} \text{INR}_{ki}.
\end{align}
\end{defi}

As it is clear, the difference between such a concept and the regular notion of an independent set lies in the fact that in the latter case, the interference between any pair of links should be below a certain threshold (e.g., noise level), whereas in the former case, the interference between all of the links is at such a low level (determined by condition (\ref{tin_opt})) that makes it (to within a constant gap) information-theoretically optimal to treat all the interference as noise. Based on the concept of ITIS, we define our scheduling scheme as follows.

\begin{defi}[ITLinQ]\label{itq}
The information-theoretic link scheduling (in short, ITLinQ) scheme is a spectrum sharing mechanism which at each time, schedules the sources in an information-theoretic independent set (ITIS) to transmit simultaneously. Moreover, all the destinations will treat their incoming interference as noise.
\end{defi}


\begin{remk}
In order to gain more intuition about the information theoretic independent sets, one can consider a simple sufficient condition for the scheduling condition in (\ref{itis_cond}). It is easy to verify that a subset of links $\mathcal{S}$ form an ITIS if for any link $i\in\mathcal{S}$,
\begin{align*}
\text{INR}_{ij}\leq\sqrt{\text{SNR}_i}~,~~ \text{INR}_{ji}\leq\sqrt{\text{SNR}_i}~,~ \forall j\in\mathcal{S}\setminus \{i\},
\end{align*}
which is the same as the condition for the optimality of TIN for a network with only two links which is mentioned in \cite{etw}. In fact, this condition compares the \emph{ratio} between the INR and SNR values \emph{in dB scale} with a fixed threshold of $\frac{1}{2}$. This is the main distinction of this condition compared to the conditions used in FlashLinQ, in which the \emph{difference} between the INR and SNR values in dB scale is compared with a fixed threshold. We will use this sufficient condition later in the paper for both the capacity analysis and the distributed implementation of the ITLinQ scheme.
\end{remk}


In Section \ref{decent}, we will show how to implement the ITLinQ scheme in a distributed way. However, for now, we will focus on characterizing the fraction of the capacity region that ITLinQ is able to achieve in a specific network setting.

\subsection{Capacity Analysis of the ITLinQ Scheme}\label{cap}

In this section, we analyze the fraction of the capacity region that the ITLinQ scheme can achieve to within a gap in a network with a large number of links.
We consider a network in which the sources are placed uniformly and independently inside a circle of radius $R$. After placing the source nodes, each destination node $\text{D}_i$ (associated with the source node $\text{S}_i$) is assumed to be located within a distance $r_n=r_0n^{-\beta}$ of $\text{S}_i$, where $r_0$ is a fixed distance and $\beta$ is a positive exponent. Note that we do not assume any particular distribution for the placement of the destination nodes as long as each of them lies within a circle of radius $r_0 n^{-\beta}$ around its corresponding source node. This represents a dense network in which the intended destinations are typically located closer to their sources, which is a characteristic of D2D networks. Moreover, in this section, we assume that each channel gain is a deterministic function of the distance between its corresponding source and destination. In fact, we consider the path-loss model for the channel gains in which the squared magnitude of the channel gain at distance $r$ is equal to $g_{0}r^{-\alpha}$, where $g_{0}\in\mathbb{R}$ is a fixed real number and $\alpha$ denotes the path-loss exponent. In Section \ref{fading}, we will study the case in which Rayleigh fading is included in the channel model as well.

For such a network and channel model, we have the following theorem (which will be proved later in this section) that presents a guarantee on the fraction of the capacity region that can be achieved by the ITLinQ scheme.

\begin{theo}\label{main_frac}
For sufficiently large number of links ($n\rightarrow\infty$) in the above model, the ITLinQ scheme can almost-surely achieve a fraction $\lambda$ of the capacity region within a gap of $k$ bits\footnote{This implies that for any rate tuple $(R_1,...,R_n)$ in the capacity region of the network, ITLinQ is (almost-surely) able to achieve the rate tuple $(\lambda R_1-k, ..., \lambda R_n-k$).}, where
%
%
\begin{align*}
\begin{cases}
\lambda= \frac{2\pi R^2}{\sqrt{3}\gamma^2} n^{\beta-1}~,~ k\leq\frac{2\pi R^2}{\sqrt{3}\gamma^2}\frac{\log 3n}{n^{1-\beta}} &\emph{if } 0<\beta<1\\
\lambda= \frac{\ln \left({\ln n}\right)}{\ln n} ~,~ k\leq \log (\ln n)+\frac{(\log 3) \ln \left({\ln n}\right)}{\ln n}&\emph{if }\beta=1\\
\lambda= \frac{1}{\left\lfloor  \tfrac{1}{\beta-1}+\tfrac{1}{2}\right\rfloor+1} ~,~ k\leq \frac{\log 3n}{\left\lfloor  \tfrac{1}{\beta-1}+\tfrac{1}{2}\right\rfloor+1}&\emph{if }\beta>1,
\end{cases}
\end{align*}

in which $\gamma=\sqrt[2\alpha]{\frac{P}{N} {g_0} r^{\alpha}_0}$ is a constant independent of $n$.
\end{theo}

\begin{remk}
The achievable fraction of the capacity region expressed in Theorem \ref{main_frac} is only a \emph{lower bound} on the fraction of the capacity region that ITLinQ is able to achieve, and therefore, ITLinQ can \emph{guarantee} the achievability of this fraction of the capacity region (to within a gap of $k$ bits).
\end{remk}

Figure \ref{thm_comp} illustrates the impact of the maximum source-destination distance decreasing rate on the fraction of the capacity region that can be achieved by ITLinQ.\footnote{Since the focus of the comparison is on the order of the fractions achievable by ITLinQ in different regimes of $\beta$, the parameters are chosen such that the constant $\frac{2\pi R^2}{\sqrt{3}\gamma^2}$ in Theorem \ref{main_frac} is equal to 1. Hence, any choice of the parameters for which $\frac{2\pi R^2}{\sqrt{3}\gamma^2}=1$ is a valid choice.} If the maximum source-destination distance is proportional to $n^{-\beta}$ such that $0<\beta <1$, then the ITLinQ scheme is capable of asymptotically achieving a fraction proportional to $\frac{1}{n^{1-\beta}}$ of the capacity region, within a vanishing gap. However, if the maximum source-destination distance scales as $n^{-1}$, then the achievable fraction of the capacity region decreases as $\frac{\ln \left({\ln n}\right)}{\ln n}$ which declines much slower than the previous case. In this case, the gap increases very slowly with respect to $n$. Finally, in the case that the maximum distance between each source and its corresponding destination scales faster than $n^{-1}$, we can achieve at least a \emph{constant} fraction of the capacity region for asymptotically large number of links which is a considerable improvement, whereas the gap is increasing with the number of links. This matches the natural intuition that the closer the destinations are located to their corresponding sources, the more the signal-to-interference-plus-noise ratio and the higher the fraction of the capacity that can be achieved by the ITLinQ scheme. Also, as a baseline, we have included the fraction of the capacity region that TDMA and independent set scheduling can achieve, which is $\frac{1}{n}$ for both schemes.\footnote{The achievable fraction of the capacity region by independent set scheduling was derived through numerical analysis.}

\begin{figure}[h]
\centering
\includegraphics[trim = 0.5in 2.8in 0.6in 2.9in, clip,width=0.45\textwidth]{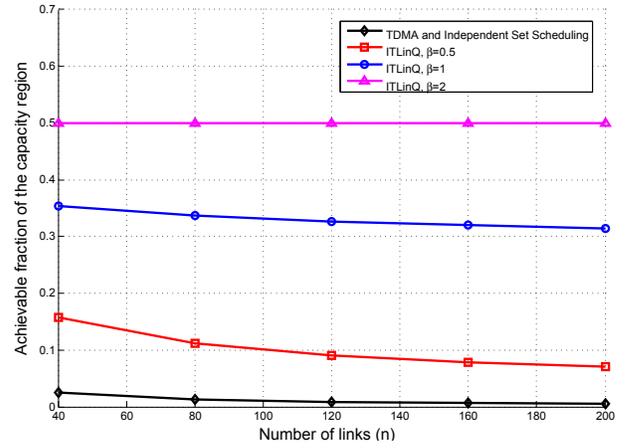}
\caption{Comparison of the guaranteed achievable fraction of capacity region by the ITLinQ scheme in different regimes with TDMA and independent set scheduling.}
\label{thm_comp}
\end{figure}


\begin{remk}
As an immediate application of Theorem \ref{main_frac}, one can consider the model in which all the $n$ source and the $n$ destination nodes are located uniformly and independently within a circular area of radius $R$, and each destination gets associated with its \emph{closest} source. For such a model, it is straightforward to show that ITLinQ can almost-surely achieve a fraction $\lambda= \frac{2\sqrt{3}\pi R^2}{3\gamma^2} n^{\beta-1}$ of the capacity region to within a gap of $k\leq\frac{2\sqrt{3}\pi R^2}{3\gamma^2}\frac{\log 3n}{n^{1-\beta}}$ for any $\beta<\frac{1}{2}$, when $n\rightarrow\infty$.
\end{remk}

\subsubsection{Proof of Theorem \ref{main_frac}}

In order to characterize the fraction of the capacity region that ITLinQ is able to achieve and prove Theorem \ref{main_frac}, we seek to find the minimum number of information-theoretic independent sets which cover all the links and we will then do time-sharing among these subsets. More precisely, if we denote the set of all the information-theoretic independent subsets of a network composed of $n$ source-destination pairs by $\mathcal{S}_n$, then we are interested in the minimum-cardinality subset of $\mathcal{S}_n$ whose members cover all the links; i.e., their union is equal to the set of all the links $\{1,...,n\}$. Denote such a subset by $\mathcal{S}^*_n$ and let $\kappa_n=|\mathcal{S}^*_n|$. We will show that time-sharing among these $\kappa_n$ information-theoretic independent sets can achieve the fractions of the capacity region mentioned in Theorem \ref{main_frac}. As the first step of the proof, we characterize the achievable fraction of the capacity region by the ITLinQ scheme and its gap with respect to the random variable $\kappa_n$ in the following lemma.

\begin{lem}\label{schfrac}
The ITLinQ scheme can achieve a fraction $\frac{1}{\kappa_n}$ of the capacity region of a network composed of $n$ source-destination pairs to within a gap of $\frac{\log 3n}{\kappa_n}$.
\end{lem}

\begin{proof}
Consider any rate tuple $(R_1,...,R_n)$ inside the capacity region of the network and consider any ITIS $\mathcal{U}\in \mathcal{S}^*_n$. From the result in \cite{tin}, since TIN is information-theoretically optimal in $\mathcal{U}$ (to within a constant gap), the rate tuple $(\bar{R}_{1,\mathcal{U}},...,\bar{R}_{n,\mathcal{U}})$ is achievable in the $\frac{1}{\kappa_n}$ fraction of time which is allocated to $\mathcal{U}$, where 
\begin{align*}
\bar{R}_{i,\mathcal{U}}=\begin{cases}
R_i-\log 3|\mathcal{U}| & i\in\mathcal{U} \\
0 & i\notin\mathcal{U}.
\end{cases}
\end{align*}

Therefore, the rate achieved by any link $i\in\{1,...,n\}$ in the network through the ITLinQ scheme, denoted by $R_{i,ITLinQ}$, can be lower bounded as
\begin{align*}
R_{i,ITLinQ}&=\frac{1}{\kappa_n} \sum_{\mathcal{U}\in\mathcal{S}^*_n} \bar{R}_{i,\mathcal{U}}\\
&=\frac{1}{\kappa_n} \sum_{\mathcal{U}\in\mathcal{S}^*_n: i\in\mathcal{U}} (R_i-\log 3|\mathcal{U}|)\\
&\geq \frac{1}{\kappa_n} (R_i-\log 3n)\numberthis\label{setcov}\\
&= \frac{1}{\kappa_n} R_i- \frac{\log 3n}{\kappa_n},
\end{align*}
where (\ref{setcov}) follows from the fact that the subsets in $\mathcal{S}^*_n$ cover all the links $\{1,...,n\}$ and that for every $\mathcal{U}\in\mathcal{S}^*_n$, we have $|\mathcal{U}|\leq n$. This completes the proof.
\end{proof}

Therefore, to find an achievable fraction of the capacity region by ITLinQ, we need to find an upper bound on $\kappa_n$, that is the minimum number of information-theoretic independent subsets which cover all of the links. One way to find such an upper bound is to restrict the TIN-optimality condition in (\ref{tin_opt}). In other words, we need to find another condition that implies condition (\ref{tin_opt}), but is more restricted and more tractable than (\ref{tin_opt}). Imposing such a restricted sufficient condition will reduce the number of information-theoretic independent subsets, hence leading to an upper bound on $\kappa_n$. To this end, we present Lemma \ref{two_indep}. In the following, we denote the distance between source $i$ and destination $j$ by $d_{\text{S}_i \text{D}_j}$ and the distance between sources $i$ and $j$ by $d_{\text{S}_i \text{S}_j}$, $\forall i,j$.

\begin{lem}\label{two_indep}
If in a network of $n$ source-destination pairs within the framework of the model in Section \ref{cap}, the distance between $\emph{S}_i$ and $\emph{S}_j$ satisfies $d_{\emph{S}_i \emph{S}_j}>\gamma n^{-\beta/2}+r_0 n^{-\beta}$, then 
\begin{align*}
\max((\emph{INR}_{ji})^2,(\emph{INR}_{ij})^2) <  \min(\emph{SNR}_i,\emph{SNR}_j).
\end{align*}
\end{lem}

\begin{proof}
Based on the model considered in Section \ref{cap}, we know that $d_{\text{S}_i \text{D}_i}\leq r_0 n^{-\beta}$ and $d_{\text{S}_j \text{D}_j}\leq r_0 n^{-\beta}$. Moreover, from the triangle inequality, we will have $d_{\text{S}_i \text{D}_j}\geq d_{\text{S}_i \text{S}_j}-d_{\text{S}_j \text{D}_j}> \gamma n^{-\beta/2}$. Similarly, we have $d_{\text{S}_j \text{D}_i}>  \gamma n^{-\beta/2}$. Therefore, we can get
\begin{align}\label{snri}
\text{SNR}_i= \frac{P}{N} g_{0}{d_{\text{S}_i \text{D}_i}}^{-\alpha}\geq \frac{P}{N} g_{0}{\left(r_0 n^{-\beta}\right)}^{-\alpha}=\frac{P}{N} g_{0}{r_0}^{-\alpha}n^{\alpha\beta},
\end{align}
and
\begin{align*}
\text{INR}_{ji}&= \frac{P}{N} g_{0}{d_{\text{S}_i \text{D}_j}}^{-\alpha}\\
&< \frac{P}{N} g_{0}{\left(\gamma n^{-\beta/2}\right)}^{-\alpha}=\frac{P}{N} g_{0}{\gamma}^{-\alpha}n^{\alpha\beta/2}.\numberthis\label{inrji}
\end{align*}

Combining (\ref{snri}) and (\ref{inrji}), we will have 
\begin{align}\label{inrsq}
(\text{INR}_{ji})^2 < \left(\frac{P}{N} g_{0}\right)^2 \gamma^{-2\alpha}n^{\alpha\beta}=
\frac{P}{N} g_{0}{r_0}^{-\alpha}n^{\alpha\beta} \leq \text{SNR}_i,
\end{align}
and likewise, we can show that 
\begin{align}\label{inrsqq}
(\text{INR}_{ij})^2 <  \text{SNR}_i.
\end{align}

Combining (\ref{inrsq}) with (\ref{inrsqq}) yields $\max((\text{INR}_{ji})^2,(\text{INR}_{ij})^2) < \text{SNR}_i$. By symmetry, we will also have $\max((\text{INR}_{ji})^2,(\text{INR}_{ij})^2) < \text{SNR}_j$. This completes the proof.
\end{proof}

Consequently, Lemma \ref{two_indep} implies that there exists a threshold distance of $d_{th,n}=\gamma n^{-\beta/2}+r_0 n^{-\beta}$ such that if the distance between two sources is greater than this threshold, the corresponding pair of links are considered to be information-theoretically independent; i.e., the interference they cause on each other is at a sufficiently low level that it is information-theoretically optimal to treat it as noise (to within a constant gap).

Therefore, given an network of $n$ source-destination pairs with nodes spread as mentioned in the model in the beginning of Section \ref{cap}, we can build a corresponding undirected graph $G_n=(V_n,E_n)$ where $V_n=\{1,...,n\}$ is the set of vertices and $(i,j)\in E_n$ if and only if $d_{\text{S}_i \text{S}_j}\leq d_{th,n}$; i.e., two nodes are connected together if and only if the distance between their sources is no larger than the threshold distance $d_{th,n}$. We call the resultant graph $G_n$ the \emph{information-theoretic conflict graph} of the original network. Clearly, this graph is a random geometric graph \cite{RGG}.

To return to our original problem, note that we needed to find an upper bound on $\kappa_n$. The following lemma provides such an upper bound.

\begin{lem}\label{chrom}
$\kappa_n\leq \chi(G_n)$, where $\chi(.)$ denotes the chromatic number.
\end{lem}

\begin{proof}
The chromatic number of $G_n$ is the smallest number of colors that can be assigned to all of the nodes of $G_n$ such that no two adjacent nodes have the same color. Therefore, considering the subsets of the links which receive the same color, $\chi(G_n)$ is the minimum number of subsets of the links which cover all the links and each of which consist of links whose sources have distance larger than $d_{th,n}$. From Lemma \ref{two_indep}, it is easy to show that if for three distinct links $i,j,k$, all the pairwise source distances are larger than $d_{th,n}$, then we will have that all the squared INR's within the subnetwork consisting of links $\{i,j,k\}$ are less than all the SNR's. Extending this argument, we can see that all the independent subsets of $G_n$ automatically satisfy the TIN-optimality condition of (\ref{tin_opt}) and hence are also information-theoretic independent subsets. Therefore, $\kappa_n$, which denotes the minimum number of information-theoretic independent subsets that cover all the links, can be no more than $\chi(G_n)$, the chromatic number of $G_n$.
\end{proof}

Thus the final step is to characterize the asymptotic distribution of $\chi(G_n)$. In this step, we will use parts (i), (ii) and (iv) of Theorem 1.1 in \cite{RGG} which we bring here for the sake of completeness. Consider a positive integer $d$ and a norm $\|.\|$ on $\mathbb{R}^d$. Suppose we have $n$ points $x_1,...,x_n$ in $\mathbb{R}^d$ and a threshold distance $r$, where $\lim_{n\rightarrow\infty} r=0$. Then, the random geometric graph $G_n$ is defined as a graph with vertex set $\{1,...,n\}$ in which vertices $i$ and $j$ are adjacent if and only if the $\|x_i-x_j\|\leq r$. Defining $f(n)\ll g(n)$ and $f(n)\sim g(n)$ to be equivalent to $\lim_{n\rightarrow\infty} \frac{f(n)}{g(n)}=0$ and $\lim_{n\rightarrow\infty} \frac{f(n)}{g(n)}=1$, respectively, we have the following theorem.

\begin{theo}[\hspace*{-.05in}\cite{RGG}]\label{th_rgg}
For the random geometric graph $G_n$, the following hold.
\begin{itemize}
\item[(i)] If $nr^d\leq n^{-\alpha}$ for some fixed $\alpha>0$, then
\begin{align*}
\hspace*{-.1in}\mathbb{P}\bigg(\chi(G_n)&\in\bigg\{\left\lfloor\left|\frac{\ln n}{\ln(nr^d)}\right|+\frac{1}{2}\right\rfloor,\left\lfloor\left|\frac{\ln n}{\ln(nr^d)}\right|+\frac{1}{2}\right\rfloor+1\bigg\}\\
&\qquad\qquad\text{for all but finitely many $n$}\bigg)=1.
\end{align*}
\item[(ii)] If $n^{-\epsilon}\ll nr^d\ll \ln n$ for all $\epsilon>0$, then
\begin{align*}
\chi(G_n)\sim \ln n / \ln \left(\frac{\ln n}{nr^d}\right) ~~a.s.
\end{align*}
\item[(iii)] If $nr^d\gg \ln n$ (but still $r\rightarrow 0$), then
\begin{align*}
\chi(G_n)\sim \frac{\emph{vol}(B)}{2^d \delta}\sigma nr^d ~~ a.s.,
\end{align*}
where $B$ is the unit ball in $\mathbb{R}^d$, $\sigma$ is the ``maximum density'' of the distribution of nodes in $\mathbb{R}^d$ and $\delta$ is the ``packing density'', defined in \cite{RGG}.
\end{itemize}
\end{theo}

Using Theorem \ref{th_rgg}, the following lemma characterizes the asymptotic behavior of the chromatic number of the information-theoretic conflict graph $G_n$.

\begin{lem}\label{chrbounds}
For the information-theoretic conflict graph $G_n$, $\chi(G_n)$ exhibits the following behavior as $n\rightarrow\infty$:
\begin{itemize}
\item If $0<\beta<1$, then $\frac{\chi(G_n)}{n^{1-\beta}}\overset{\text{a.s.}}{\longrightarrow} \frac{\sqrt{3}}{2\pi R^2}\gamma^2$.

\item If $\beta=1$, then $\frac{\chi(G_n)}{{{\ln n}}/{{\ln \left({\ln n}\right)}}}\overset{\text{a.s.}}{\longrightarrow} 1$.

\item If $\beta>1$, then
\begin{align*}
\mathbb{P}\bigg(&\chi(G_n)  {\longrightarrow} \left\lfloor  \tfrac{1}{\beta-1}+\tfrac{1}{2}\right\rfloor \\
&\text{ or } \chi(G_n)  {\longrightarrow} \left\lfloor  \tfrac{1}{\beta-1}+\tfrac{1}{2}\right\rfloor+1\bigg)=1.
\end{align*}
\end{itemize}

\end{lem}

\begin{proof}
Since the information-theoretic conflict graph $G_n$ is a random geometric graph with threshold distance $d_{th,n}=\gamma n^{-\beta/2}+r_0 n^{-\beta}$ and the nodes are distributed in $\mathbb{R}^2$, we can  make use of Theorem \ref{th_rgg}. We will have the following cases:

\begin{itemize}
\item If $0<\beta<1$, then $nd_{th,n}^2=\gamma^2 n^{1-\beta}+r_0^2 n^{1-2\beta}\gg \ln n$, and therefore we can use part (iii) of Theorem \ref{th_rgg}. Note that the dominant term in $\gamma^2 n^{1-\beta}+r_0^2 n^{1-2\beta}$ is the first term, since $\beta>0$. Also, as mentioned in \cite{RGG}, for the case of Euclidean norm in $\mathbb{R}^2$, we have $\delta=\frac{\pi}{2\sqrt{3}}$ and $\text{vol}(B)=\pi$. Also, since the distribution of the nodes is uniform on a circle of radius $R$, we have $\sigma=\frac{1}{\pi R^2}$. Therefore, we can get $\frac{\chi(G_n)}{n^{1-\beta}}\overset{\text{a.s.}}{\longrightarrow} \frac{\sqrt{3}}{2\pi R^2}\gamma^2$.

\item If $\beta=1$, then $nd_{th,n}^2=\gamma^2+r_0^2 n^{1-2\beta}$ which converges to a constant asymptotically, since $1-2\beta<0$. This enables us to use part (ii) of Theorem \ref{th_rgg}, since $n^{-\epsilon}\ll\gamma^2+r_0^2 n^{1-2\beta}\ll \ln n$ for all $\epsilon>0$. Therefore, we have $\frac{\chi(G_n)}{{{\ln n}}/{{\ln \left({\ln n}\right)}}}\overset{\text{a.s.}}{\longrightarrow} 1$.

\item If $\beta>1$, then $nd_{th,n}^2=\gamma^2 n^{-(\beta-1)}+r_0^2 n^{-(2\beta-1)}$, where $2\beta-1>\beta-1>0$. Thus, we can make use of part (i) of Theorem \ref{th_rgg} to get
\begin{align*}
\mathbb{P}\bigg(&\chi(G_n)  {\longrightarrow} \left\lfloor  \tfrac{1}{\beta-1}+\tfrac{1}{2}\right\rfloor \\
&\text{ or } \chi(G_n)  {\longrightarrow} \left\lfloor  \tfrac{1}{\beta-1}+\tfrac{1}{2}\right\rfloor+1\bigg)=1.
\end{align*}
\end{itemize}
\end{proof}

The proof of Theorem \ref{main_frac} then follows immediately from Lemmas \ref{schfrac}, \ref{chrom} and \ref{chrbounds} and also the fact that the continuous function $f(x)=\frac{1}{x}$ preserves almost-sure convergence (continuous mapping theorem \cite{as}).






\subsubsection{Impact of Rayleigh Fading on the Capacity Analysis}\label{fading}

One of the most important phenomena in wireless networks is the concept of channel fading. Even though fading seems to be a detrimental aspect of wireless networks, it can also be helpful if it is viewed in a careful way. Probably the most well-known example for this is receive diversity at multi-antenna receivers, where we can make use of independently faded signals to combine them in the best way, leading to an improvement in the received SNR.

Hence, it would be interesting to figure out how fading can affect the results we derived so far on the fraction of the capacity region that ITLinQ can achieve. In this section, we focus on this problem, considering the same model for the spatial location of the nodes as in Section \ref{cap} with the difference that here, we consider the squared magnitude of the channel gain at distance $r$ to be $g_0 r^{-\alpha}$ where $g_0$ represents the Rayleigh fade of the channel modeled as an exponential random variable with normalized mean of 1. We consider a slow fading scenario (i.e., block fading), where the rate of change of the channel characteristics is much smaller than the rate of change of the transmitted signal. Hence, the channel fade $g_0$ remains fixed during the transmission within each block of communication (which corresponds to a scheduling phase of ITLinQ) and changes i.i.d from one block to the next.

The definition of information-theoretic independent set (ITIS) still remains the same as before, i.e., within each block of communication with revisited channel gain values (modeled as $g_0r^{-\alpha}$), a subset $\mathcal{S}\subseteq\{1,...,n\}$ in which for any link $i\in\mathcal{S}$ condition (\ref{itis_cond}) is satisfied is an ITIS in that block. Obviously, introducing Rayleigh fading into the channel model adds another source of randomness in the analysis of the fraction of the capacity region achieved by ITLinQ, which is due to the dependence of ITIS's on the random fade of the channels.

 
However, we can still make use of Lemma \ref{schfrac} to characterize the fraction of the capacity region that ITLinQ can achieve in each block of communication when Rayleigh fading is also included in the channel model. In this case, the faded interference may no longer be Gaussian, but we can circumvent this issue due to the recent result in \cite{worstnoise}. In \cite{worstnoise}, the authors show that in a multi-user network, Gaussian noise is the worst-case additive noise in the sense that any rate tuple that can be achieved under the assumption of Gaussian noise can also be achieved under non-Gaussian additive noise of the same variance. 

Therefore, by treating the aggregate (non-Gaussian) noise plus faded interference at each destination as a Gaussian noise, we achieve a lower bound on the achievable rate of ITLinQ. As a result, Lemma \ref{schfrac} would still hold in a fading scenario, meaning that in each block of communication ITLinQ can achieve a fraction $\frac{1}{\kappa_n}$ of the capacity region to within a gap of $\frac{\log 3n}{\kappa_n}$, where now $\kappa_n$, the minimum number of ITIS's whose union contains all the links, depends both on the spatial location of the links and the realization of the fading. Characterizing the distribution of $\kappa_n$ in the fading scenario (even in the asymptote of $n \rightarrow \infty$) is quite challenging, hence we will use numerical evaluation in the rest of this section to analyze the average fraction of the capacity that ITLinQ achieves in the fading scenario (i.e., $\mathbb{E}\left[\frac{1}{\kappa_n}\right]$).

We consider the same network model of Section \ref{cap} for the placement of the nodes (in which source nodes are distributed uniformly within a circle of radius $R$ and each destination node is located within a distance $r_0 n^{-\beta}$ of its corresponding source node) and we evaluate the average fraction of the capacity region that ITLinQ is able to achieve (to within a gap) for both cases of with and without Rayleigh fading. The result is illustrated is Figure \ref{fracs_fading}.
\begin{figure}[h]
\centering
\includegraphics[trim = 0in 2.3in 0.2in 2.4in, clip,width=0.48\textwidth]{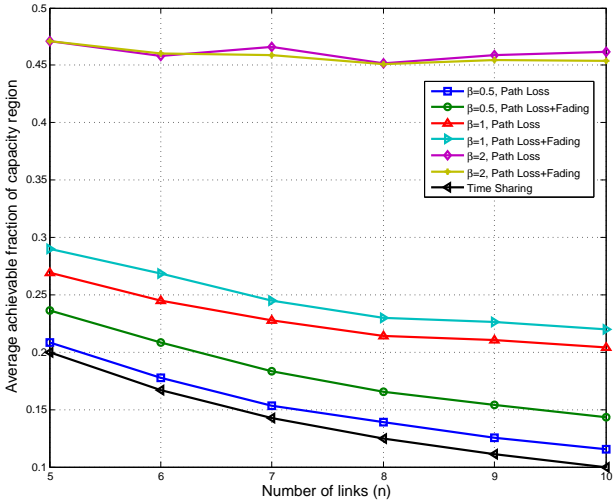}
\caption{Average achievable fraction of the capacity region by ITLinQ scheme with and without fading and comparison with time-sharing.}
\label{fracs_fading}
\end{figure}

By considering Figure \ref{fracs_fading}, we can now compare the average fraction of the capacity region that ITLinQ achieves in the fading and non-fading scenario (both with the same average channel gains). We interestingly note that there is an improvement in the case where Rayleigh fading is also included (in particular, for $\beta \leq 1$). The intuition behind this improvement can be explained as follows. Consider the ITIS condition (\ref{itis_cond}) rewritten in the following form
\begin{align}\label{fdg_cond}
g_{ii} d_{\text{S}_i \text{D}_i}^{-\alpha} \geq \frac{P}{N} \max_{j\in\mathcal{S}\setminus \{i\}} g_{ij} d_{\text{S}_j \text{D}_i}^{-\alpha} \max_{k\in\mathcal{S}\setminus \{i\}} g_{ki} d_{\text{S}_i \text{D}_k}^{-\alpha},
\end{align}
where $\forall i,j$, $g_{ij}$ is the exponential fading random variable of the channel between source $j$ and destination $i$. For fixed power and noise levels and spatial distribution of the nodes in the network, condition (\ref{fdg_cond}) reveals the opportunity that fading is providing in this case. In fact, there are specific locations of nodes in the network for which condition (\ref{fdg_cond}) cannot be satisfied in a deterministic path-loss setting. However, our numerical results show that the randomness due to the inclusion of Rayleigh fading can help this condition to be satisfied for more subsets of the links, resulting in an improvement in the achievable fraction of the capacity region. This can, therefore, be viewed as another case where fading is helpful in terms of the system performance. For the case of $\beta>1$, since the destination nodes get very close to their corresponding source nodes, interference is already at a very low level and therefore fading cannot be of much help and may even degrade the performance by a small amount, as depicted in Figure \ref{fracs_fading}.

Finally, using Lemma 1, we can also quantify the gap to the fraction of the capacity region that ITLinQ is able to achieve in each block of communication to be $\frac{\log 3n}{\kappa_n}$. For the above network model, we numerically evaluate and plot the average gap (i.e., $\mathbb{E}\left[\frac{\log 3n}{\kappa_n}\right]$)in Figure \ref{gaps}.
\begin{figure}[h]
\centering
\includegraphics[trim = 0in 2.3in 0.2in 2.4in,clip,width=0.48\textwidth]{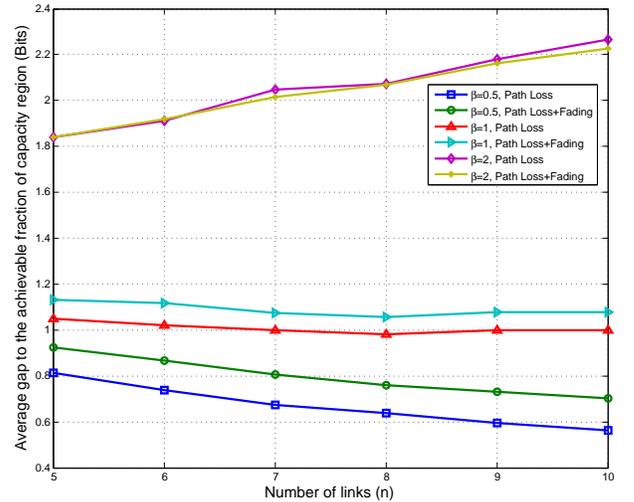}
\caption{Comparison of the average gap to the achievable fraction of the capacity region by ITLinQ with and without fading.}
\label{gaps}
\end{figure}
Similar to the non-fading scenario of Section \ref{cap}, we note that, for the case of $\beta \leq 1$, the average gap again does not scale with the size of the network and always remains less than 1.2 bits, independent of the number of links. However, for the case of $\beta>1$ where a constant fraction of the capacity region can be achieved asymptotically, the gap increases with the number of links and Theorem \ref{main_frac} predicts that the increase is logarithmic with respect to $n$.

\section{A Distributed Method for Implemenation of ITLinQ}\label{decent}

In this section, we present a distributed algorithm for putting the ITLinQ scheme into practice in real-world networks. The algorithm is inspired by the FlashLinQ distributed algorithm \cite{FLQ} and its complexity is exactly at the same level as the FlashLinQ algorithm. However, as we will demonstrate through numerical analysis in Section \ref{comp}, it significantly outperforms FlashLinQ in a certain network scenario.

As mentioned in Section \ref{defs}, we consider wireless networks consisting of $n$ source-destination pairs. In each execution of the algorithm, to address the issue of fairness among the links, we first permute the links randomly and reindex them from 1 to $n$ based on the realization of the random permutation, as also done in \cite{FLQ}. This new indexing of the links corresponds to a priority order of the links: link $i$ has higher priority than link $j$ if $i<j$, $\forall i,j\in\{1,...,n\}$. Then, link 1 is always scheduled to transmit at the current time frame and for the remaining links, each link is scheduled if it does not cause and receive ``too much'' interference to and from the higher priority links. The conditions for defining the level of ``too much'' interference for link $j\in\{2,...n\}$ are as follows, where $\eta$ is a design parameter:

\begin{itemize}
\item At $\text{D}_j$, the following conditions must be satisfied:
\begin{align}\label{itqdist1}
\text{INR}_{ji}\leq \text{SNR}^\eta_j,~\forall i<j,
\end{align}
which imply that $\text{D}_j$ does not receive too much interference from higher-priority links.

\item At $\text{S}_j$, the following conditions must be satisfied:
\begin{align}\label{itqdist2}
\text{INR}_{ij}\leq \text{SNR}^\eta_j,~\forall i<j,
\end{align}
which imply that $\text{S}_j$ does not cause too much interference at higher-priority links.
\end{itemize}

As it is clear, there are two major differences here with respect to the FlashLinQ scheduling conditions: The first difference is that instead of considering the raw fraction $\text{SIR}=\frac{\text{SNR}}{\text{INR}}$, here we are considering an exponent for the SNR term, which is completely inspired by the condition for the optimality of TIN (\ref{tin_opt}). The second difference is that in condition (\ref{itqdist2}), the outgoing interference of each link is compared to \emph{its own} SNR rather than other links' SNR's. This is also inspired by the TIN-optimality condition (\ref{tin_opt}).

In fact, if the parameter $\eta$ is set to $\eta=0.5$, then conditions (\ref{itqdist1}) and (\ref{itqdist2}) imply that the TIN-optimality condition (\ref{tin_opt}) is satisfied at link $j$. This means that link $j$ can safely be added to the information-theoretic independent subset of higher priority links and get scheduled to transmit in the current time frame. This algorithm, therefore, seeks to find the largest possible \emph{distributed} information-theoretic independent subset based on the priority ordering of the links. However, it is clear that selecting $\eta=0.5$ might be too pessimistic and restrictive, and may prevent some links which cause and receive low levels of interference from being scheduled. Therefore, we will leave this variable as a design parameter, and as we will see in Section \ref{comp}, tuning this parameter can indeed improve the achievable sum-rate by this scheduling algorithm.

The remaining question is: How can the sources and destinations check whether their pertinent conditions are satisfied? This can be done by a simple signaling mechanism which is inspired by the FlashLinQ algorithm \cite{FLQ} and is a two-phase process, in each of which we assume that each link uses its own frequency band and transmissions are interference-free:

\begin{itemize}
\item In the first phase, all the sources transmit signals at their full power $P$. The destinations will receive their own desired signals and also all the interfering signals in separate frequency bands. Then, the destinations estimate their received $\text{SNR}$'s and $\text{INR}$'s and check if their desired conditions (\ref{itqdist1}) are satisfied. This phase is the same as that of FlashLinQ \cite{FLQ}.

\item In the second phase, contrary to the ``inverse power echo'' mentioned in the FlashLinQ algorithm \cite{FLQ}, the destinations also transmit signals at the same power level $P$ of the sources. Similar to the first phase, in this phase all the sources can estimate the value of their desired $\text{SNR}$'s and $\text{INR}$'s in order to verify the validity of condition (\ref{itqdist2}).
\end{itemize}

\begin{remk}
Clearly, power control at the transmitters may lead to an improvement in the performance of the scheme. However, due to the complication in implementing power control among the links in a distributed way, we disregard it in our scheme and use full power at all the transmitters. See e.g. \cite{power} on power control algorithms in D2D underlaid cellular networks.
\end{remk}

As it is obvious, the complexity of our distributed signaling mechanism is completely comparable to that of the FlashLinQ algorithm.

\section{Numerical Analysis}\label{sim}

In this section, we numerically analyze the performance of distributed ITLinQ in two distinct settings. First, in Section \ref{thm_sim} we assess the performance of distributed ITLinQ under the model developed in Section \ref{cap}. Furthermore, in Section \ref{comp} we compare the performance of distributed ITLinQ with FlashLinQ in a model similar to the one considered in \cite{FLQ}.

For concreteness, the algorithm used in this section for the performance evaluation of distributed ITLinQ is illustrated in pseudo-code format in Algorithm \ref{dist_itlq}. Here, we assume that through multiple iterations of the training mechanism introduced in Section \ref{decent}, each link is aware of the values of its own SNR and all its incoming and outgoing INR's and it is also aware of the active higher-order links. We assume that the knowledge of the active higher-order links is also available to each link in implementing the FlashLinQ scheme.
\begin{algorithm}
\caption{Implementation of Distributed ITLinQ}\label{dist_itlq}
\begin{algorithmic}[1]
\State \textbf{initialize} $active(1)=1,~active(j)=0,~\forall j=2,...,n$;
\State \textbf{for} $j=2,...,n$
\State \quad $S_j=\{i: i\leq j\text{ and }active(i)=1\}$
\State \quad \textbf{if} $\text{INR}_{ji}\leq M\text{SNR}_j^\eta$ at $\text{D}_j$, $\forall i\in S_j$
\State \quad \quad $flag_{\text{D}_j}$=1;
\State \quad \textbf{endif}
\State \quad \textbf{if} $\text{INR}_{ij}\leq M\text{SNR}_j^\eta$ at $\text{S}_j$, $\forall i\in S_j$
\State \quad \quad $flag_{\text{S}_j}$=1;
\State \quad \textbf{endif}
\State \quad $active(j)=flag_{\text{D}_j}~.~ flag_{\text{S}_j}$;
\State \textbf{end}
\State \textbf{return} $active$
\end{algorithmic}
\end{algorithm}
For the implementation of ITLinQ distributively, we also consider a second tuning parameter $M$ which adds more flexibility to our scheme. This parameter can in general be tuned to optimize the performance of the algorithm in any network setting. For the results of this section, We will set $M$ to be equal to 25 dB. Algorithm \ref{dist_itlq} returns a vector $active$ of length $n$ which specifies whether or not each link should be scheduled. In particular, for any $j\in\{1,...,n\}$, link $j$ is scheduled if and only if $active(j)=1$.

\subsection{Performance of Distributed ITLinQ under the Model of Section \ref{cap}}\label{thm_sim}

In this section, we analyze the performance of distributed ITLinQ under the model of Section \ref{cap}. In Section \ref{cap} we characterized the fraction of the capacity region that ITLinQ can achieve (to within a gap) with full knowledge of all the channel gains for an asymptotically large number of links. One may wonder how well ITLinQ can perform under this model for \emph{finite} number of links where each link only has a \emph{local} knowledge of the channel gain values. This provides the motivation for studying distributed ITLinQ in this setting.

For the sake of numerical analysis in this section, we assume that the sources are distributed uniformly within a circle of radius $R=10 km$ and each destination is assumed to be uniformly distributed in a circle of radius $r_n=r_0 n^{-\beta}$ around its corresponding source node where $r_0$ is chosen to be equal to $1km$. The squared channel gain at distance $r$ is taken to be equal to $r^{-2.5}$ (the constant $g_0$ is assumed to be normalized to 1). The transmit power is taken to be 10 dBm and the additive white Gaussian noise variance at the destination nodes is set to -110 dBm.

We have evaluated the sum-rate achievable by distributed ITLinQ in this setting with the value of $\eta$ set to 0.5 and $\beta$ taking values 0.5, 1 and 2. The result is illustrated in Figure \ref{thm_comp_sim}.
\begin{figure}[h]
\centering
\includegraphics[trim = 0.05in 2.3in 0in 2.3in, clip,width=0.47\textwidth]{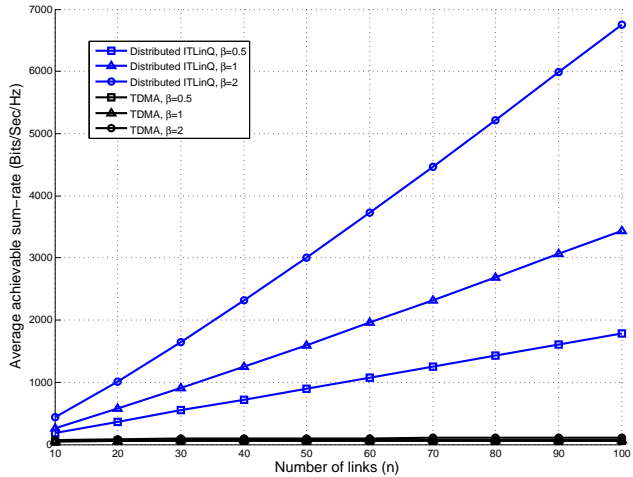}
\caption{Performance comparison of distributed ITLinQ with time-sharing under the model of Section \ref{cap}.}
\label{thm_comp_sim}
\end{figure}
For the sake of comparison, we have also plotted the average sum-rate that can be achieved by time-sharing among the links. As the figure demonstrates, there is a huge sum-rate improvement by using distributed ITLinQ over time-sharing. Moreover, while our theoretical analysis in Theorem \ref{main_frac} characterizes the fraction of the capacity region achievable by fully-centralized ITLinQ, what we observe in Figure \ref{thm_comp_sim} is that the network can still enjoy the significant sum-rate improvement that a distributed implementation of ITLinQ can provide.

\subsection{Performance Comparison of the distributed ITLinQ and FlashLinQ}\label{comp}

In this section, we will illustrate the performance of our distributed algorithm and compare it with FlashLinQ through numerical analysis. We drop $n$ links randomly in a $1km\times 1km$ square. The length of each link, which is the distance between its corresponding source and destination, is taken to be a uniform random variable in the interval $[2,65m]$. As in \cite{FLQ}, we use the carrier frequency of 2.4 GHz and a bandwidth of 5 MHz. The noise power spectral density is considered to be -184 dBm/Hz. The transmit power is set to 20 dBm. Moreover, the channel follows the LoS model in ITU-1411. In particular, if the base station antenna height is denoted by $h_b$, the mobile station antenna height is denoted by $h_m$ and the transmission wavelength is denoted by $\lambda$, then the transmission loss (in dB) at distance $d$ is taken to be equal to
\begin{align*}
L=L_{bp}+6+\begin{cases}
20\log_{10} \left(\frac{d}{R_{bp}}\right) & \text{if } d\leq R_{bp}\\
40\log_{10} \left(\frac{d}{R_{bp}}\right) & \text{if } d> R_{bp}
\end{cases},
\end{align*}
where $R_{bp}=\frac{4 h_b h_m}{\lambda}$ denotes the breakpoint distance and $L_{bp}=\left|20\log_{10} \left(\frac{\lambda^2}{8\pi h_b h_m}\right)\right|$ denotes the basic transmission loss at the break point. As in \cite{FLQ}, we assume all the antenna heights to be equal to $1.5m$, alongside with a log-normal shadowing with standard deviation of 10 dB. The antenna gain per device is taken to be -2.5 dB and the noise figure is assumed to be 7 dB.

Figure \ref{comparison} demonstrates the sum-rate achievable by the distributed ITLinQ scheme for different values of $\eta$ and its comparison to FlashLinQ. The implementation of FlashLinQ follows the same steps as in \cite{FLQ} and in particular, the threshold values $\gamma_{TX}$ and $\gamma_{RX}$ are taken to be equal to 9 dB.
\begin{figure}[h]
\centering
\subfigure[]{
\includegraphics[trim = 0.4in 2.6in 0.2in 2.7in, clip,width=0.497\textwidth]{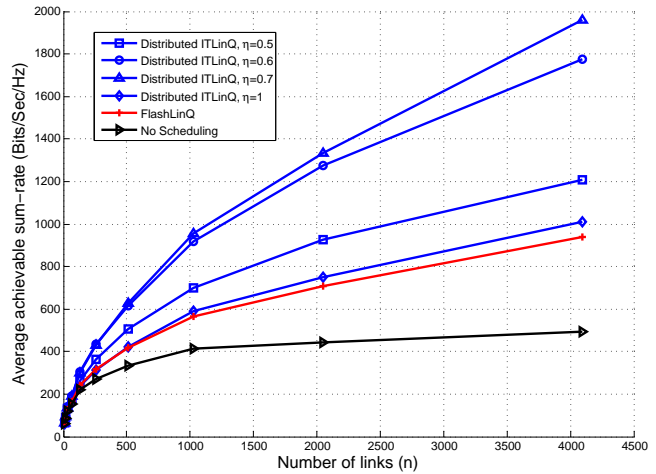}
\label{comparison}}
\subfigure[]{
\includegraphics[trim = 0.615in 2.72in .75in 2.9in, clip,width=0.46\textwidth]{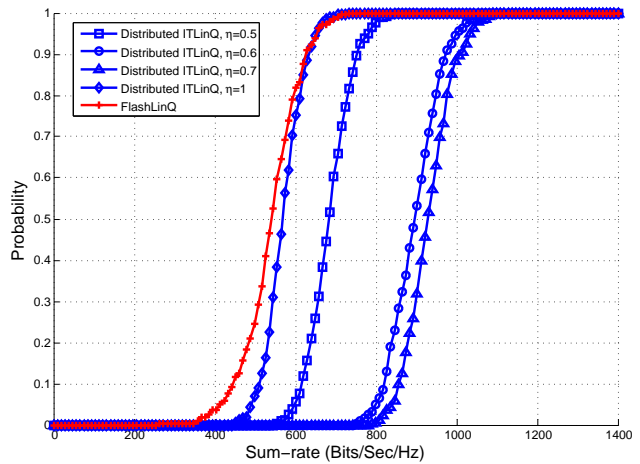}
\label{cdf}}
\caption{\subref{comparison} Comparison of the sum-rate performance of distributed ITLinQ, FlashLinQ and the no-scheduling case, and \subref{cdf} Comparison of the cumulative distribution function of the average link rate achieved by distributed ITLinQ and FlashLinQ in a network of 1024 links.}
\end{figure}
%
As the figure illustrates, tuning the parameter $\eta$ can lead to considerable gains over FlashLinQ. For the case of $\eta=0.5$, in which conditions (\ref{itqdist1}) and (\ref{itqdist2}) are sufficient for the optimality of TIN (to within a constant gap), distributed ITLinQ exhibits over 28\% gain compared to FlashLinQ for 4096 links. Interestingly, setting $\eta=0.7$ results in more than 110\% gain over FlashLinQ for 4096 links. However, as we increase $\eta$ to 1, more and more links get scheduled which results in a degradation in the overall performance.
As a baseline, the achievable sum-rate when there is no scheduling (i.e., all the links operate simultaneously) is also plotted in Figure \ref{comparison}.

Moreover, in the same setting, we also study the cumulative distribution function (CDF) of the sum-rate in a network of 1024 links. The result is depicted in Figure \ref{cdf}.
%
%
Again, the same trend occurs in this plot, showing that distributed ITLinQ, especially for the value of $\eta=0.7$, can result in considerable uniform gain compared to the sum-rate achievable by FlashLinQ. For instance, with 50\% probability, the sum-rate achieved by FlashLinQ is less than 540 bits/sec/Hz while with the same probability, the sum-rate achieved by distributed ITLinQ is less than 928 bits/sec/Hz.

Another natural aspect of distributed scheduling schemes that is of considerable importance is the  issue of fairness among the links. In particular, the scheduling scheme should take care of all links fairly, regardless of them being strong or weak. It can be seen that the distributed ITLinQ scheme favors strong links more than weak links. To highlight this issue, a network with two links AB and CD is shown in Figure \ref{ex_unfair}.
\begin{figure}[h]
\centering
\includegraphics[trim = 2.5in 3in 2.5in 2.7in, clip,width=0.35\textwidth]{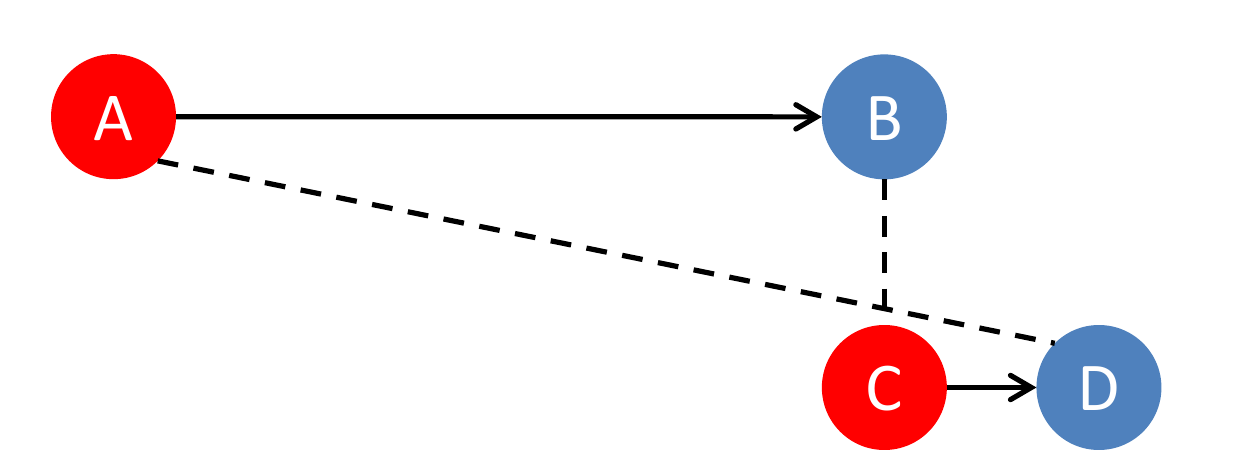}
\caption{An example of the case where distributed ITLinQ might be unfair.}
\label{ex_unfair}
\end{figure}
In this figure, link AB is a low-SNR link and link CD is a high-SNR link. Moreover, Destination node B suffers from strong interference due to the source node C. To see why ITLinQ may be unfair in such a scenario regardless of the priority of the links, consider the following two cases:
\begin{itemize}
\item If link AB has a lower priority than link CD, link CD is first scheduled. Then, destination B checks its scheduling condition $\text{INR}_{BC}\leq M \text{SNR}_{BA}^\eta$ and with a high probability may find that it is not satisfied (since the interference from C is strong compared to the signal power received from A). This will prevent link AB from being scheduled.

\item If link AB has a higher priority than link CD, it will be scheduled first. Then, since both destination node D is receiving a low amount of interference from A (compared to the signal power from D) and source node C is causing a low amount of interference at B (compared to the signal power it delivers to D), link CD will also get scheduled and hurts the transmission of link AB.
\end{itemize}

Therefore, in both cases, the low-SNR link AB will not get a high rate, if any. This motivates a modification of the distributed ITLinQ scheme to account for this issue. To this end, we present a \emph{fair} version of distributed ITLinQ as follows. Inspired by the example in Figure \ref{ex_unfair}, in the fair ITLinQ algorithm, the high-SNR links should get scheduled in a more restrictive way. This can be done by decreasing the parameters $\eta$ (and $M$) in the scheduling condition for the outgoing interference of high-SNR links. In general, $\eta$ (and $M$) need to be a descending function of SNR. However, one simple solution would be to choose a threshold $\text{SNR}_{th}$ such that if the SNR of a link is higher than this threshold, $\eta$ and $M$ are altered to decreased values $\bar{\eta}$ and $\bar{M}$. The pseudo-code for the fair ITLinQ scheme is presented in Algorithm \ref{fair_itlinq}.
\begin{algorithm}
\caption{Fair ITLinQ}\label{fair_itlinq}
\begin{algorithmic}[1]
\State \textbf{initialize} $active(1)=1,~active(j)=0,~\forall j=2,...,n$;
\State \textbf{for} $j=2,...,n$
\State \quad $S_j=\{i: i\leq j\text{ and }active(i)=1\}$
\State \quad \textbf{if} $\text{INR}_{ji}\leq M\text{SNR}_j^\eta$ at $\text{D}_j$, $\forall i\in S_j$
\State \quad \quad $flag_{\text{D}_j}$=1;
\State \quad \textbf{endif}
\State \quad \textbf{if} $\text{SNR}_j\leq \text{SNR}_{th}$
\State \quad \quad \textbf{if} $\text{INR}_{ij}\leq M\text{SNR}_j^\eta$ at $\text{S}_j$, $\forall i\in S_j$
\State \quad \quad \quad $flag_{\text{S}_j}$=1;
\State \quad \quad \textbf{endif}
\State \quad \textbf{else}
\State \quad \quad \textbf{if} $\text{INR}_{ij}\leq \bar{M}\text{SNR}_j^{\bar{\eta}}$ at $\text{S}_j$, $\forall i\in S_j$
\State \quad \quad \quad $flag_{\text{S}_j}$=1;
\State \quad \quad \textbf{endif}
\State \quad \textbf{endif}
\State \quad $active(j)=flag_{\text{D}_j}~.~ flag_{\text{S}_j}$;
\State \textbf{end}
\State \textbf{return} $active$
\end{algorithmic}
\end{algorithm}
To assess the performance of fair ITLinQ in terms of fairness, we have numerically evaluated the CDF of the link rates (averaged over both priorities and locations) for a network with 1024 links under the same model as the one mentioned in the beginning of this Section. The threshold value for high-SNR is chosen to be $\text{SNR}_{th}=110$ dB and the modified parameters are set to $\bar{\eta}=0.6$ and $\bar{M}=20$ dB. Figure \ref{cdf_fair} compares the CDF of the average link rate by distributed ITLinQ (with $\eta=0.7$), fair ITLinQ and FlashLinQ.
\begin{figure}[h]
\centering
\includegraphics[trim = .8in 2.8in .9in 3in, clip,width=0.45\textwidth]{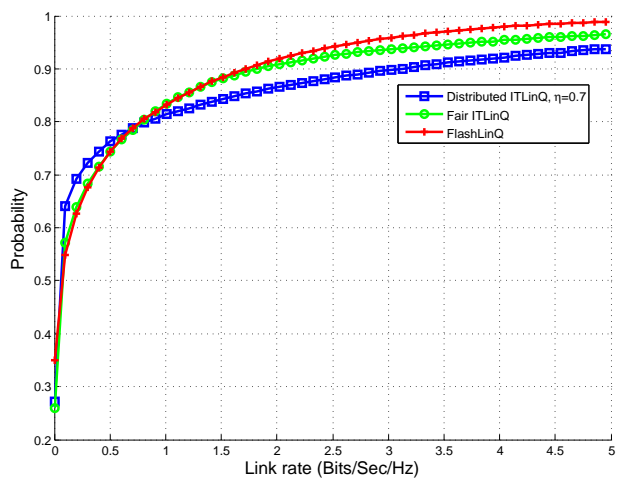}
\caption{Comparison of the average link rate CDF of distributed ITLinQ, fair ITLinQ and FlashLinQ for a network with 1024 links.}
\label{cdf_fair}
\end{figure}
As the figure illustrates, fair ITLinQ can improve the tail distribution of distributed ITLinQ and perform as well as FlashLinQ in terms of fairness. This certainly does not come for free and in fact, there is a trade-off between fairness and the achievable sum-rate. The sum-rate achieved by fair ITLinQ is compared with FlashLinQ in Figure \ref{sumrate_fair}.
\begin{figure}[h]
\centering
\includegraphics[trim = .7in 2.8in .9in 3in, clip,width=0.45\textwidth]{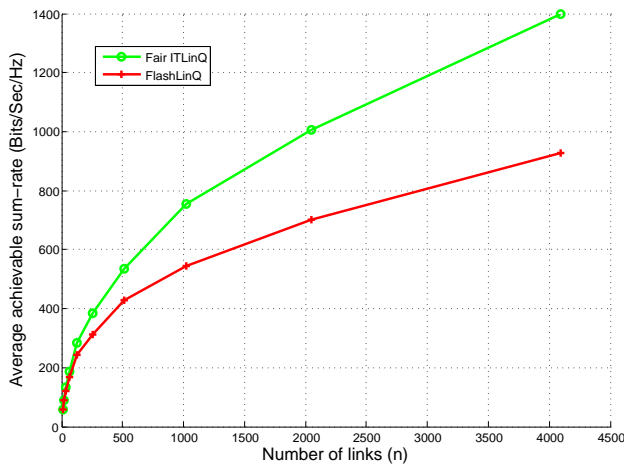}
\caption{Comparison of the sum-rate achievable by fair ITLinQ and FlashLinQ.}
\label{sumrate_fair}
\end{figure}
As the figure illustrates, for the case of 4096 links, the sum-rate gain of fair ITLinQ over FlashLinQ is more than 50\%. For more information regarding the software implementation of ITLinQ, the reader is referred to \cite{itlq_web}.

\section{Concluding Remarks and Future Directions}\label{conc}
In this paper, we introduced a new spectrum sharing scheme, called information-theoretic link scheduling (ITLinQ), in order to manage the interference in wireless networks. The scheme relies on the recently-found optimality condition for treating interference as noise and at each time, schedules a subset of links in which treating interference as noise is information-theoretically optimal (to within a constant gap). We presented a performance guarantee of the ITLinQ scheme by characterizing the fraction of the capacity region that it is able to achieve in a specific network setting. Moreover, we developed a distributed way of implementing the ITLinQ scheme and showed, via numerical analysis, that it yields considerable gains over FlashLinQ, a similar recently-proposed scheduling algorithm. We also showed how to address the issue of fairness among the links in the network by introducing a fair version of the distributed ITLinQ scheme.

There are multiple future directions to consider following this work. First, in this work we presented a lower bound on the fraction of the capacity region that ITLinQ is able to achieve. It would be interesting to see if there exists an upper bound on the achievable fraction of the capacity region by the ITLinQ scheme so that its performance can be characterized more precisely. This can be accomplished through new outer bounds on the capacity region of the interference channels. Second, it is worthy to study the performance of the ITLinQ scheme under different models of D2D networks other than the ones presented in this work. In particular, models with time-varying topology that allow links to dynamically share the spectrum might be of interest. Third, one can think about generalizing the ITLinQ scheme to multihop D2D networks. Especially, the result in \cite{kkk} shows that coupling between interference management and relaying strategies can provide significant gains. Hence, it would be worth figuring out the performance improvement that ITLinQ is able to guarantee in a multihop setting. Fourth, it would be interesting to investigate the impact of more advanced interference management techniques, such as successive interference cancellation, on ITLinQ. For example, recent results in \cite{localview,tim,topology} demonstrate that by a careful use of repetition coding at the transmitters and temporal interference neutralization at the receivers, one can achieve spectral efficiency gains that are considerably beyond the common interference avoidance approach. Thus, an interesting future direction can be to characterize the impact of structured repetition coding and temporal interference neutralization on ITLinQ. Finally, as mentioned in Section \ref{intro}, FlashLinQ has been implemented in a practical testbed and its performance has been measured and evaluated. Hence, it is interesting to also test the ITLinQ scheme and observe the improvements that it is able to yield in practice.

\begin{IEEEbiography}
    [{\includegraphics[width=1in,height=1.25in,clip,keepaspectratio]{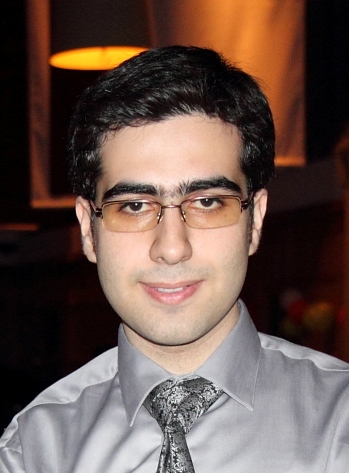}}]{Navid NaderiAlizadeh}
received his B.S. degree in electrical engineering from Sharif University of Technology, Tehran, Iran, in 2011 and his M.S. degree in electrical and computer engineering from Cornell University, Ithaca, NY, in 2014. He joined University of Southern California, Los Angeles, CA, in 2014, where he is pursuing his Ph.D. in the department of electrical engineering. His research is focused on developing spectrum sharing mechanisms in wireless networks, especially in device-to-device (D2D) communication networks. Navid obtained the first rank in the Iranian nationwide university entrance exam in 2007. He was also a recipient of Jacobs scholarship in 2011.
\end{IEEEbiography}

\begin{IEEEbiography}
    [{\includegraphics[width=1in,height=1.25in,clip,keepaspectratio]{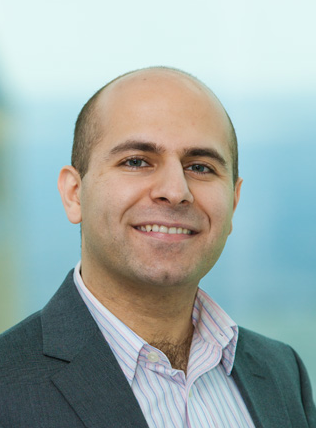}}]{Amir Salman Avestimehr}
received the B.S. degree in electrical engineering from Sharif University of Technology, Tehran, Iran, in 2003 and the M.S. degree and Ph.D. degree in electrical engineering and computer science, both from the University of California, Berkeley, in 2005 and 2008, respectively. He is currently an Associate Professor at the EE department of University of Southern California, Los Angeles, CA. He was also a postdoctoral scholar at the Center for the Mathematics of Information (CMI) at the California Institute of Technology, Pasadena, in 2008. His research interests include information theory, the theory of communications, and their applications.

Dr. Avestimehr has received a number of awards for research and teaching, including the Communications Society and Information Theory Society Joint Paper Award in 2013, the Presidential Early Career Award for Scientists and Engineers (PECASE) in 2011, the Michael Tien 72 Excellence in Teaching Award in 2012, the Young Investigator Program (YIP) award from the U. S. Air Force Office of Scientific Research in 2011, the National Science Foundation CAREER award in 2010, and the David J. Sakrison Memorial Prize in 2008. He is currently an Associate Editor for the IEEE Transactions on Information Theory. He has also been a Guest Associate Editor for the IEEE Transactions on Information Theory Special Issue on Interference Networks and General Co-Chair of the 2012 North America Information Theory Summer School and the 2012 Workshop on Interference Networks.
\end{IEEEbiography}


\begin{thebibliography}{9}

\bibitem{dyspan}
N. Naderializadeh and A. S. Avestimehr, ``ITLinQ: A New Approach for Spectrum Sharing,'' in proceedings of \emph{IEEE International Symposium on Dynamic Spectrum Access Networks (DySPAN)}, pp. 327-333,  April 2014.

\bibitem{isit}
N. Naderializadeh and A. S. Avestimehr, ``ITLinQ: A New Approach for Spectrum Sharing in Device-to-Device Communication Systems,'' to appear in proceedings of \emph{IEEE International Symposium on Information Theory}, June-July 2014.

\bibitem{ericsn_5g}
``5G Radio Access - Research and Vision,'' Ericsson white paper, June 2013.

\bibitem{d2djef}
X. Lin, J. Andrews, and A. Ghosh, ``Spectrum Sharing for Device-to-Device Communication in Cellular Networks,'' e-print arXiv:1305.4219.


\bibitem{d2dlte1}
L. Lei, Z. Zhong, C. Lin, and X. Shen, ``Operator controlled device-to-device communications in LTE-advanced networks,'' \emph{IEEE Wireless Communications}, vol. 19, no. 3, pp. 96-104, June 2012.

\bibitem{d2ddopp}
K. Doppler, M. Rinne, C. Wijting, C. Ribeiro, and K. Hugl, ``Device-to-Device Communication as an Underlay to LTE-Advanced Networks,” \emph{IEEE Communications Magazine}, vol. 47, no. 12, pp. 42-49, Dec. 2009.

\bibitem{fodor}
G. Fodor, E. Dahlman, G. Mildh, S. Parkvall, N. Reider, G. Mikl\'{o}s, and Z. Tur\'{a}nyi, ``Design Aspects of Network Assisted Device-to-Device Communications,'' \emph{IEEE Communications Magazine}, vol. 50, no. 3, pp. 170-177, Mar. 2012.


\bibitem{d2dinterf}
M. Zulhasnine, C. Huang, and A. Srinivasan, ``Efficient Resource Allocation for Device-to-Device Communication Underlaying LTE Network,'' in proceedings of \emph{IEEE 6th International Conference on Wireless and Mobile Computing, Networking and Communications (WiMob)},  pp. 368-375, Oct. 2010.

\bibitem{d2dicufn}
Q. Duong and O. Shin, ``Distance-Based Interference Coordination for Device-to-Device Communications in Cellular Networks,'' in proceedings of \emph{Fifth International Conference on Ubiquitous and Future Networks (ICUFN)}, pp. 776-779, July 2013.

\bibitem{caching}
M. Ji, G. Caire, and A. F. Molisch, ``Wireless Device-to-Device Caching Networks: Basic Principles and System Performance,'' e-print arXiv:1305.5216.

\bibitem{FLQ}
X. Wu, S. Tavildar, S. Shakkottai, T. Richardson, J. Li, R. Laroia, and A. Jovicic, ``FlashLinQ: A Synchronous Distributed Scheduler for Peer-to-Peer Ad Hoc Networks,'' \emph{IEEE/ACM Transactions on Networking}, vol. 21, no. 4, pp. 1215-1228, Aug. 2013.

\bibitem{is1}
L. Tassiulas and A. Ephremides, ``Jointly Optimal Routing and Scheduling in Packet Radio Networks,'' \emph{IEEE Transactions on Information Theory}, vol. 38, no. 1, pp. 165-168, Jan. 1992.

\bibitem{is2}
P. Chaporkar, K. Kar, and S. Sarkar, ``Throughput Guarantees Through Maximal Scheduling in Wireless Networks,'' in proceedings of \emph{43rd Annual Allerton Conference on Communications, Control, and Computing}, 2005.

\bibitem{is3}
X. Lin, and N. B. Shroff, ``The Impact of Imperfect Scheduling on Cross-Layer Rate Control in Multihop Wireless Networks,'' in proceedings of \emph{24th Annual Joint Conference of the IEEE Computer and Communications Societies (INFOCOM)}, 2005.

\bibitem{is4}
G. Sharma, R. R. Mazumdar, and N. B. Shroff, ``Maximum Weighted Matching with Interference Constraints,'' in proceedings of \emph{Fourth Annual IEEE International Conference on Pervasive Computing and Communications Workshops, (PerCom)}, 2006.

\bibitem{is5}
G. Sharma, R. R. Mazumdar, and N. B. Shroff, ``On the Complexity of Scheduling in Wireless Networks,'' in proceedings of \emph{The Twelfth Annual ACM International Conference on Mobile Computing and Networking (MobiCom)}, 2006.

\bibitem{kumar}
P. Gupta and P. R. Kumar, ``The Capacity of Wireless Networks,'' \emph{IEEE Transactions on Information Theory}, vol. 46, no. 2, pp. 388-404, Mar. 2000.

\bibitem{tin}
C. Geng, N. Naderializadeh, A. S. Avestimehr, and S. Jafar, ``On the
Optimality of Treating Interference as Noise,'' e-print arXiv:1305.4610.

\bibitem{hk}
T. S. Han and K. Kobayashi, ``A New Achievable Rate Region for the Interference Channel'', \emph{IEEE Transactions on Information Theory}, vol. 27, no. 1, pp. 49-60, Jan. 1981.


\bibitem{etw}
R. H. Etkin, D. N. C. Tse, and H. Wang, ``Gaussian Interference Channel Capacity to Within One Bit,'' \emph{IEEE Transactions on Information Theory}, vol. 54, no. 12, pp. 5534-5562, Dec. 2008.


\bibitem{cadambe}
V. R. Cadambe and S. A. Jafar, ``Interference Alignment and Degrees of Freedom of the $K$-User Interference Channel,'' \emph{IEEE Transactions on Information Theory}, vol. 54, no. 8, pp. 3425-3441, Aug. 2008.

\bibitem{maddahali}
M. A. Maddah-Ali, A. S. Motahari, and A. K. Khandani, ``Communication Over MIMO X Channels: Interference Alignment, Decomposition, and Performance Analysis,'' \emph{IEEE Transactions on Information Theory}, vol. 54, no. 8, pp. 3457-3470, Aug. 2008.

\bibitem{bpt}
G. Bresler, A. Parekh, and D. N. C. Tse, ``The Approximate Capacity of the Many-to-One and One-to-Many Gaussian Interference Channels,'' \emph{IEEE Transactions on Information Theory}, vol. 56, no. 9, pp. 4566-4592, Sep. 2010.

\bibitem{oen}
O. Ordentlich, U. Erez, and B. Nazer, ``The Approximate Sum Capacity of the Symmetric Gaussian $K$-User Interference Channel,'' e-print arXiv:1206.0197.

\bibitem{jv}
A. Jafarian and S. Vishwanath, ``Achievable Rates for $K$-User Gaussian
Interference Channels,'' \emph{IEEE Transactions on Information Theory}, vol. 58, no. 7, pp. 4367-4380, July 2012.

\bibitem{adt}
A. S. Avestimehr, S. N. Diggavi, and D. N. C. Tse, ``Wireless Network Information Flow: A Deterministic Approach,'' \emph{IEEE Transactions on Information Theory}, vol. 57, no. 4, pp. 1872-1905, Apr. 2011.



\bibitem{RGG}
C. Mcdiarmid and T. M\"{u}ller, ``On the Chromatic Number of Random Geometric Graphs,'' \emph{Combinatorica}, vol. 31, no. 4, pp. 423-488, Nov. 2011.

\bibitem{as}
H. B. Mann and A. Wald, ``On the Statistical Treatment of Linear Stochastic Difference Equations,'' \emph{Econometrica} , vol. 11, no. 3/4, pp. 173-220, 1943.

\bibitem{worstnoise}
I. Shomorony and A. S. Avestimehr, ``Worst-Case Additive Noise in Wireless Networks,'' \emph{IEEE Transactions on Information Theory}, vol. 59, no. 6, pp. 3833-3847, June 2013.

\bibitem{power}
N. Lee, X. Lin, J. G. Andrews, and R. W. Heath Jr, ``Power Control for D2D Underlaid Cellular Networks: Modeling, Algorithms and Analysis,'' e-print arXiv:1305.6161.

\bibitem{itlq_web}
\url{http://www-scf.usc.edu/~naderial/ITLinQ.html}.

\bibitem{kkk}
I. Shomorony and A. S. Avestimehr, ``Degrees of Freedom of Two-Hop Wireless Networks: Everyone Gets the Entire Cake,'' \emph{IEEE Transactions on Information Theory}, vol. 60, no. 5, pp. 2417-2431, May 2014.

\bibitem{localview}
V. Aggarwal, A. S. Avestimehr, and A. Sabharwal, ``On Achieving Local View Capacity Via Maximal Independent Graph Scheduling,'' \emph{IEEE Transactions on Information Theory}, vol. 57, no. 5, pp. 2711-2729, May 2011.

\bibitem{tim}
S. A. Jafar, ``Topological Interference Management through Index Coding,'' e-print arXiv:1301.3106.

\bibitem{topology}
N. Naderializadeh and A. S. Avestimehr, ``Interference Networks with No CSIT: Impact of Topology,'' e-print arXiv:1302.0296.


\end{thebibliography}
\end{document}